%% file: QPL2014preproc.tex
\documentclass[preliminary,copyright,creativecommons,noncommercial]{eptcs}
\providecommand{\event}{QPL 2014}

\input{preamble.tex}

\title{Terminality implies non-signalling}  
\author{Bob Coecke 
  \institute{University of Oxford}
  \email{bob.coecke@cs.ox.ac.uk}  
  }
\date{}
\def\titlerunning{Terminality equals non-signalling}     
\def\authorrunning{B. Coecke}   

\begin{document} 
\maketitle
\begin{abstract} 
A `process theory' is any theory of systems and processes which admits sequential and parallel composition.  `Terminality' 
unifies normalisation of pure states, trace-preservation of CP-maps, and adding up to identity of positive operators in quantum theory, and generalises this to arbitrary process theories.  We show that  terminality and non-signalling coincide in any process theory, provided one makes causal structure explicit.  In fact, making causal structure explicit is  necessary to even make sense of non-signalling in process theories.  We conclude that because of its much simpler mathematical form, terminality should  be taken to be a more fundamental notion than non-signalling. 




\end{abstract}

\section{Introduction}

Causality related notions are prominent in many areas of physics, and the relationships between these are by no means obvious.  Examples that are relevant to us are:
\bit
\item[{\bf C1}] Relativistic space-time  is often abstracted as a partial ordering, called \em causal structure \em \cite{penroseorder}. 
\item[{\bf C2}] In quantum information, for example in the context of generalised probabilistic theories \cite{Barrett}, one often relies on the notion of  \em non-signalling\em, by means of which one intends to implement these relativistic constraints for spatially distributed information-processing devices.
\item[{\bf C3}] In quantum foundations, an axiom called \em causality \em has recently been  put forward in \cite{Chiri1, Chiri2}.  For process theories \cite{JTF, CKBook} this axiom is equivalent to the mathematical notion of \em terminality\em,\footnote{Process theories are symmetric monoidal categories and causality boils down to terminality of the tensor unit.} and here we will adopt this terminology  in order to avoid confusing multiple uses of the term `causal(ity)'.\footnote{{\bf C1}, {\bf C2} as well as {\bf C3} are ofter referred to as causality.}   
\eit 
These three notions do not straightforwardly match each other. For example, in \cite{Chiri1, Chiri2} the causality axiom stands for non-signalling-from-the future (i.e.~time-like), while the above mentioned notion of non-signalling concerns  space-like  separation. Also, in \cite{CoeckeLal} it was shown that while relativistic space-like separation is invariant under time-reversal,  this is by no means the case for non-signalling. 

Here we investigate how these  notions are related.  Firstly, we argue that to even make sense of non-signalling for general process theories (cf.~{\bf C2}) one needs to make causal structure  (cf.~{\bf C1}) explicit. Then, the resulting definition of non-signalling, for the particular case of two systems,   becomes equivalent to terminality (cf.~{\bf C3}). In the process of doing so we also resolve the seeming contradiction in \cite{CoeckeLal}.



Since terminality 
is mathematically way simpler than non-signalling, it should be taken to be the more fundamental notion. %
Such a stance is already  adopted in \cite{CKBook} where it is shown that within the context of process theories  much of the relevant structure of quantum theory directly follows from terminality.  Terminality also  yields a covariance theorem \cite{CRCaucat} (building further on earlier work in \cite{Markopoulou, BIP}). 

\paragraph{Related work.} Pearl  resolved several paradoxes in probability theory \cite{PearlBook} by making causal structure explicit within probability theory, an achievement for which he received the Turing Award.  This treatment of probability theory also resulted in a fine-grained analysis of Bell's theorem by Wood-Spekkens \cite{WoodSpekkens}, exploiting many tool's of Pearl's approach.  Very recently, two other papers appeared which included very similar results, by Fritz \cite{FritzII} and Henson-Lal-Pusey \cite{HLP}. A first difference is that the  results of these authors  were obtained  within a more restrictive framework, but exploit more general scenarios.   A more fundamental difference is that these authors fix one particular causal structure and study the resulting correlations.  Here we will universally quantify over all causal structures that leave the two parties separate.  Our goal is indeed not to obtain a result with respect to a fixed causal structure, but relate terminality and non-signalling for any causal structure that is allowed for.

 
\section{Process theories with discarding}  
  
Here, by a \em  process theory \em \cite{JTF, CKBook} we mean a collection of \em systems\em , represented by \em wires\em, and \em processes\em, represented by \em boxes \em with wires as inputs and outputs.  Moreover, when we plug these boxes together: 
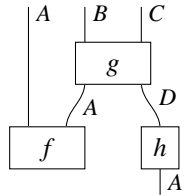
\[
\InputIfFileExists{compound-process.tikz}{}{\input{./figures/compound-process.tikz}}
\]
the resulting \em diagram \em should also be a process. For the purposes of this paper, outputs should be connected to inputs, and   diagrams should never contain causal loops. 

\begin{remark}
Equivalently,  one can say that a process theory is a symmetric monoidal category. 
\end{remark}

\begin{remark}
In the above diagram we used labels to distinguish distinct systems  and distinct wires, but we could also have done this by relying on different kinds of wires:
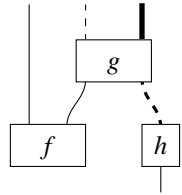
\[
\InputIfFileExists{compound-process-mw.tikz}{}{\input{./figures/compound-process-mw.tikz}}
\]
From now on we will omit labels or different kinds of wires, but all wires (within one diagram) may be interpreted as distinct. 
\end{remark}

A \em state \em is a process without inputs, and an \em effect \em is a process without outputs.  We will  assume that for each system there exists a \em discarding \em effect, which we denote as follows:
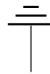
\[
\InputIfFileExists{disc.tikz}{}{\input{./figures/disc.tikz}}
\]

\begin{example}
When viewing probability theory as a process theory,  boxes are stochastic maps, states are probability distributions, and discarding is marginalization.  
\end{example}

\begin{example}
When viewing quantum theory as a process theory,  boxes include CP-maps as well as measurements and classical data processes, states include density operators, and discarding is either the trace or deletion of classical dada.  A detailed description of quantum theory as a process theory is in \cite{CK}.
\end{example}

\section{Terminality vs.~non-signalling} 

\begin{definition}
A process theory is \em terminal \em if for every process $f$ we have: 
\beq\label{eq:term}
\InputIfFileExists{norm.tikz}{}{\input{./figures/norm.tikz}} 
\eeq
\end{definition}

Terminality has a clear operational intuition: performing an operation on a system, and then discarding the resulting system, is the same as discarding that system from the start.

\begin{proposition}\label{prop:termTFAE} 
For a process theory with discarding processes, TFAE:
\bit
\item[(a)] it is terminal, 
\item[(b)] all effects are discarding, and,
\item[(c)] for each system there is only one effect.
\eit
\end{proposition}

\begin{remark}
When taking a process theory to be a symmetric monoidal category, by Proposition \ref{prop:termTFAE} (c) it follows that  `terminality' is a shorthand for `the tensor unit being terminal'. 
\end{remark}

\begin{example}
In the case of probability theory, terminality means stochasticity, which in the particular case of probability distributions means summing up to $1$, i.e.~normalisation. 
\end{example} 

Non-signalling is a bit more involved. The idea behind it is that for two spatially separated parties, say Alice and Bob,  when they share a device (= a process $f$) with two inputs and two outputs,  each having access to one input and one output, then, from one's own input-output pair one should not be able to derive the other party's input. Otherwise, this device would enable the parties to signal to each other  while they are space-like separated, hence violating relativity.   So from one of the party's perspective, say Bob's, who has no access to Alice's output, something that we can represent by discarding that output, Alice's input must not affect Bob's input-output relationship.  That is, from Bob's perspective Alice's  input can be discarded, resulting in some process $h$ between Bob's input and Bob's output.   This leads to one equation with an existential quantifier.  A second equation concern's Alice's perspective on things. 

\begin{definition}\label{def:nonsig}
A two-input-two-output process $f$ is  \em non-signalling \em if we have: 
\[
\exists h\ :  \ \ %
\InputIfFileExists{nosig1.tikz}{}{\input{./figures/nosig1.tikz}} \qquad \mbox{and}\qquad \exists h'\ :  \ \ %
\InputIfFileExists{nosig2.tikz}{}{\input{./figures/nosig2.tikz}} 
\]
\end{definition}

However,  it is a lot more delicate to say that a process theory is non-signalling.  Simply saying `a process theory is non-signalling if all its processes are' doesn't work.  For example, most process theories would include the swap process:
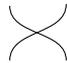
\[
\InputIfFileExists{swap.tikz}{}{\input{./figures/swap.tikz}} 
\]
which evidently violates non-signalling.  The obvious reason being that this process consists of two clear  signalling channels.  The kind of thing non-signalling aims to forbid is that without having such explicit signalling channels, signalling should not be possible.  This is exactly the kind of thing that can be said in terms of a causal structure (see below), so we will need to reconsider the notion of non-signalling in the context of causal structure, in order to make sense of a non-signalling process theory.

\begin{remark}  
One could argue that there is no reason why in the case of non-signalling the discarding effect should be unique, something which we implicitly assumed in our notation.  One could indeed weaken the definition of non-signalling for a process $f$ to: 
\[
\exists e, h\ :  \ \ %
\InputIfFileExists{nosig3.tikz}{}{\input{./figures/nosig3.tikz}} \qquad \mbox{and}\qquad \exists h', e'\ :  \ \ %
\InputIfFileExists{nosig4.tikz}{}{\input{./figures/nosig4.tikz}}  
\]
This would not affect the main claims in this paper, in that we would still be able to derive non-signalling from terminality, and hence, that terminality should be taken to be the more fundamental notion.  
\end{remark}

 
\section{Process theories with causal structure} 
 
A \em causal structure \em is a partial ordering, which we can represent by the corresponding  Hasse-diagram.  For example, if we have $\bot < a, b < \top$ then this yields the diagram:
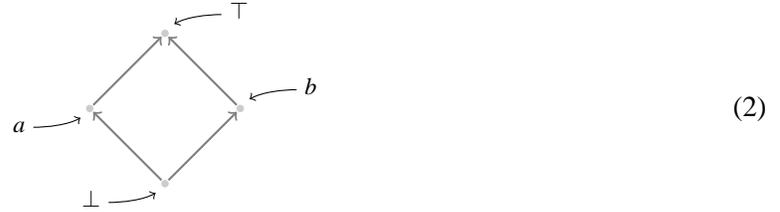
\beq\label{eq:causstruc}
\InputIfFileExists{Causaly1.tikz}{}{\input{./figures/Causaly1.tikz}}
\eeq
Here, $\bot$ can signal to $a$ and $b$, which themselves can signal to $\top$, and by transitivity, $\bot$ can also signal to $\top$.  But $a$ cannot signal to $b$ and vice versa.   

We can use such a causal structure as a support for a diagram within a process theory in the following manner, which (more or less) generalises  how \em causal networks \em are defined by Pearl in \cite{PearlBook}:\footnote{We say `more or less' since for the purposes of this paper we modify things a bit as compared to \cite{PearlBook}.}  
\bit
\item processes are positioned on nodes of the causal structure,
\item wires follow the edges, from outputs to inputs, and
\item we allow open inputs and open outputs at nodes.\footnote{It is here that we modify Pearl's approach by allowing `open inputs and outputs'.}
\eit
With respect to the  causal structure (\ref{eq:causstruc}), we could for example have:
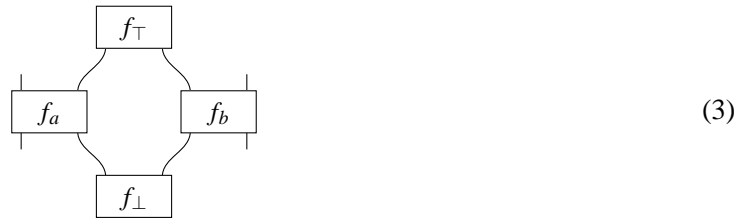
\beq\label{eq:causproc}
\InputIfFileExists{diamond.tikz}{}{\input{./figures/diamond.tikz}}
\eeq
where we allowed at $a$ and $b$ for there to be open inputs and outputs.

\begin{remark}
The manner in which a diagram of processes carries a causal structure was also considered in \cite{CRCaucat}.  In brief, each diagram already carries a shadow of a causal structure in that sequential composition can be interpreted as `after', i.e.~time-like separated, and that parallel composition can be interpreted as `while', i.e.~\em possibly \em space-like. So the only required additional specification is to state which processes are \em genuinely \em space-like separated.  We refer the reader to  \cite{CRCaucat} for more details.
\end{remark}

The type of scenario that we considered when defining non-signalling involves two parties that are not supposed to be able to have signalling channels to each other,  just like $a$ and $b$ in (\ref{eq:causstruc}).\footnote{So in the context of  \cite{CRCaucat}, these two parties have to be genuinely space-like separated.}
That Alice and Bob each have an input and an output is then made explicit by the \em causal process network \em  (\ref{eq:causproc}).  

In fact, this is the most general situation that we need to consider for the purpose of evaluating non-signalling. Indeed, while one could of course have a causal structure like the following ones:
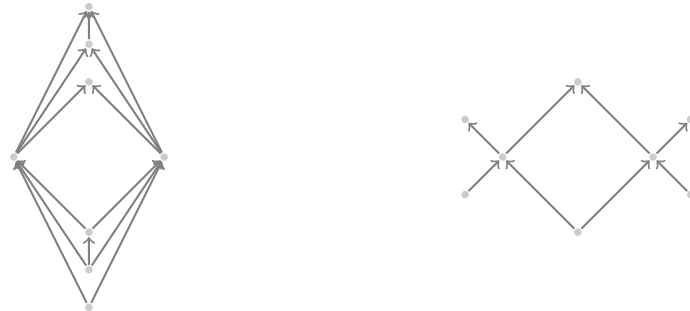
\[
\InputIfFileExists{Causaly2.tikz}{}{\input{./figures/Causaly2.tikz}}
\]
these are still covered by the diamond-shaped process network simply by treating clusters of nodes as single processes:
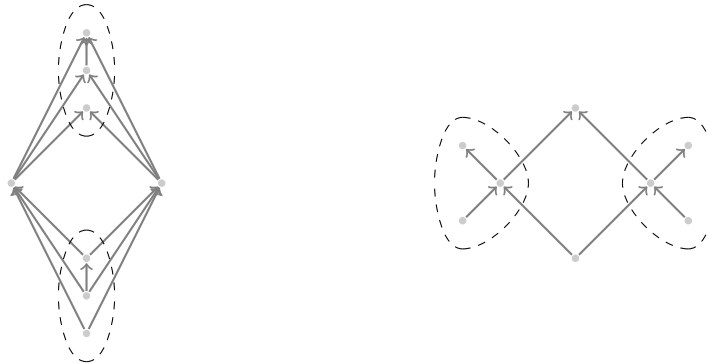
\[
\InputIfFileExists{Causaly3.tikz}{}{\input{./figures/Causaly3.tikz}}
\]

\begin{definition}
A process theory is non-signalling if all  processes of the form (\ref{eq:causproc}) are.
\end{definition}

\begin{example}[two-sorted process theories] 
When one is concerned with non-signalling, then typically, there are two kinds of systems involved. One of them would correspond to probabilities, while the other one could involve some  exotic physical systems, which may be existing e.g.~quantum theory, or hypothetical e.g.~Barrett's box world \cite{Barrett} which enables one to realise correlations between systems way beyond quantum correlations, and in fact, is extremal as far as non-signalling is concerned.  One treats the inputs and outputs at Alice's and Bob's ends as probabilities, but the `internal wiring' of the box as in Definition \ref{def:nonsig} may involve the exotic systems.  The causal process network (\ref{eq:causproc}) would then become:
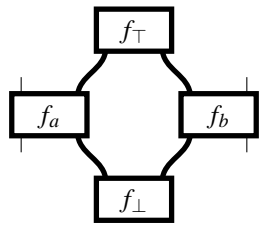
\ctikzfig{diamondb}
where the bold wires and boxes mean `exotic' while the normal wires mean `probability'.  
\end{example}

\begin{example}[local hidden variables] 
The notion of a local  hidden variable theory means that such a two-sorted diagram can be replaced by a one-sorted one, only involving normal wires, and bold then standing for `quantum'.  As we shall see below in the proof of Theorem \ref{thm:main1}, assuming terminality, which one does in quantum theory, the diamond shape immediately becomes a V-shape, and the  definition for a local hidden variable representation then becomes, given a quantum scenario involving $f_a, f_b, f_\bot$:
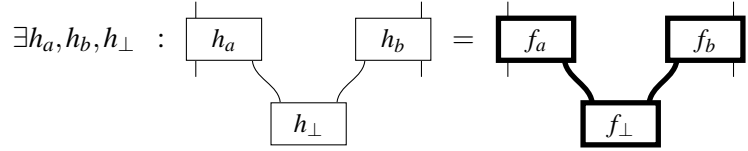
\[
\exists h_a, h_b, h_\bot \ \ :\ \ %
\InputIfFileExists{diamondV1.tikz}{}{\input{./figures/diamondV1.tikz}} \ \ = \ \ %
\InputIfFileExists{diamondV2.tikz}{}{\input{./figures/diamondV2.tikz}}
\]
 \end{example}

 \section{Main result} 


Let us first summarise what we have established so far.  We want to compare the notions of terminality and non-signalling for process theories, but noted that non-signalling of a process of the form:
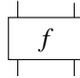
\[
\InputIfFileExists{box.tikz}{}{\input{./figures/box.tikz}}
\]
already requires explicit internal causal structure in order to define what it means for a process theory to be non-signalling, so that we can exclude internal signalling channels  such as e.g.:  
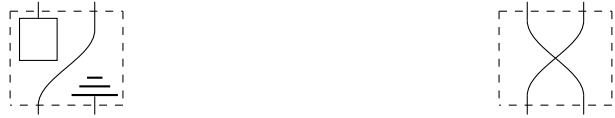
\[
\InputIfFileExists{box1.tikz}{}{\input{./figures/box1.tikz}}
\]
We concluded that the internal causal structure should be a diamond so that  $f$ should be of the form:
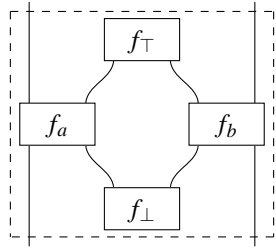
\[
\InputIfFileExists{diamonddots.tikz}{}{\input{./figures/diamonddots.tikz}}
\]
This now allows for a definition of a non-signalling process theory. 
 
\begin{theorem}\label{thm:main1}
If a process theory is terminal then it is non-signalling.
\end{theorem}
\begin{proof}
Assuming terminality, we have for $f_\top$  in (\ref{eq:causproc}):
\ctikzfig{xxx}
so (\ref{eq:causproc}) becomes:
\ctikzfig{diamondr1}
and applying terminality to $f_a$ the equational requirement  for non-signalling is now indeed obeyed: 
 \ctikzfig{diamondr2}
 where the dashed box identifies $h$ as in Definition \ref{def:nonsig}. 
\end{proof}

To prove the converse, we need one extra assumption.  What is a process with no inputs and no outputs?  It interacts with nothing, so it should be independent of anything else that happens. 

\begin{definition}
By (!) we mean that there is a unique diagram with no inputs nor outputs, the empty one.  
\end{definition} 

\begin{remark}
The justification for there only being one box with no input and no output is that it represents \em certainty\em, and in this paper all processes are conceived as happening with certainty.  For example, in the context of quantum theory this means that measurements are considered `as a whole', that is, taking into account all possible outcomes together. This is different from, for example, a diagram for teleportation such as \cite{Kindergarten}: 
 \ctikzfig{tele}
where a particular measurement outcome is considered that only happen with a probability ${1\over 4}$.
\end{remark}

\begin{theorem}
Assuming (!), if a process theory is non-signalling then it is terminal. 
\end{theorem}
\begin{proof}
Equation (\ref{eq:term})   follows from non-signalling of processes, simply by taking the input and the output  of one of the parties in Definition \ref{def:nonsig} to be trivial (i.e.~no input nor output). Since $h$ now has no inputs nor outputs it is the empty diagram by (!).   
\end{proof}

\section{Discussion}

So we achieved our goal and established equivalence of terminality and non-signalling for two parties, where we conceive the fact that terminalty implies non-signalling as the most significant result.  Terminality is both conceptually and formally the simpler and more elegant notion, and should therefore be taken to be the more fundamental one.  Some other consequences are: 
\bit
\item Much attention has been given to the notion of non-signalling in several frameworks for theories more general than quantum theory, and maybe, these frameworks should be reconsidered.
\item While in \cite{Chiri1, Chiri2} terminality was taken to mean non-signalling from the future, we showed here that it does much more than just that, also implying non-signalling in the usual sense.
\item The reason why our result demystifies the result of  \cite{CoeckeLal}
is that one simply should not try to match causal structure with non-signalling.  Causal structure is only one ingredient when defining non-signalling, and while causal structure admits a clear notion of time-reversal, the other ingredient, the process theory does not admit such a thing, since it is governed by a manifestly time-asymmetric principle such as terminality.    
\eit  

\section*{Acknowledgement}

Rob Spekkens has  for quite a while already emphasised the importance  for quantum foundations of causal structure as in causal networks. 
In chats last summer in Benasque he indicated that this may provide the key to demystifying \cite{CoeckeLal}, which, as we demonstrated here, it indeed did. In the same chat, he also strongly voiced his concerns about the the notion of non-signalling as in quantum information, which hopefully, here we (at least in part) also dethroned.   

The QPL referees provided some very useful feedback on the submitted draft and Tobias Fritz and Raymond Lal where helpful by carefully examining the relationship of our work to theirs.

\providecommand{\urlalt}[2]{\href{#1}{#2}}
\providecommand{\doi}[1]{doi:\urlalt{http://dx.doi.org/#1}{#1}}


%
%
%
%
%
%
%
%

\end{document}

%% file: preamble.tex
\usepackage{rotating}
\usepackage{floatpag}
\rotfloatpagestyle{empty}

\usepackage{hyperref}

\usepackage{amsmath,amsthm,amssymb}
\usepackage{xspace,enumerate,color,epsfig}
\usepackage{graphicx}
\graphicspath{{.}{./figures/}}
\usepackage{marvosym}

\usepackage{tikzfig}
\usepackage{stmaryrd}
\usepackage{docmute}
\usepackage{keycommand}

\input{defs.tex}

\input{tikzstyles.tex}

\let\olddagger\dagger
\renewcommand{\dagger}{\ensuremath{\olddagger}\xspace}

\usepackage{makeidx}
\makeindex

\theoremstyle{definition}
\newtheorem{theorem}{Theorem}[section]

\newtheorem{proposition}[theorem]{Proposition}

\newtheorem{definition}[theorem]{Definition}

\newtheorem{example}[theorem]{Example}

\newtheorem{example*}[theorem]{Example*}
\newtheorem{examples*}[theorem]{Examples*}
\newtheorem{remark}[theorem]{Remark}
\newtheorem{remark*}[theorem]{Remark*}

\usepackage{color}
\def\bR{\begin{color}{red}} 
\def\bB{\begin{color}{blue}}
\def\bM{\begin{color}{magenta}}
\def\bC{\begin{color}{cyan}}
\def\bW{\begin{color}{white}}
\def\bBl{\begin{color}{black}} 
\def\bG{\begin{color}{green}}
\def\bY{\begin{color}{yellow}}
\def\e{\end{color}\xspace}
\newcommand{\bit}{\begin{itemize}}
\newcommand{\eit}{\end{itemize}\par\noindent}
\newcommand{\ben}{\begin{enumerate}}
\newcommand{\een}{\end{enumerate}\par\noindent}
\newcommand{\beq}{\begin{equation}}
\newcommand{\eeq}{\end{equation}\par\noindent}
\newcommand{\beqa}{\begin{eqnarray*}}
\newcommand{\eeqa}{\end{eqnarray*}\par\noindent}
\newcommand{\beqn}{\begin{eqnarray}}
\newcommand{\eeqn}{\end{eqnarray}\par\noindent}



%% file: defs.tex
\newkeycommand{\pointmap}[style=point][1]{\,\tikz{\node[style=\commandkey{style}] (x) at (0,-0.3) {$#1$};\node [style=none] (3) at (0, 1.0) {};\node [style=none] (3) at (0, -1.0) {};\draw (x)--(0,0.7);}\,}

\newkeycommand{\typedkpoint}[style=kpoint,edgestyle=][2]{%
\,\begin{tikzpicture}
  \begin{pgfonlayer}{nodelayer}
    \node [style=none] (0) at (0, 1) {};
    \node [style=\commandkey{style}] (1) at (0, -0.75) {$#2$};
    \node [style=right label] (2) at (0.25, 0.5) {$#1$};
  \end{pgfonlayer}
  \begin{pgfonlayer}{edgelayer}
    \draw [\commandkey{edgestyle}] (1) to (0.center);
  \end{pgfonlayer}
\end{tikzpicture}\,}

\newkeycommand{\typedkpointdag}[style=kpoint adjoint,edgestyle=][2]{%
\,\begin{tikzpicture}
  \begin{pgfonlayer}{nodelayer}
    \node [style=none] (0) at (0, -1) {};
    \node [style=\commandkey{style}] (1) at (0, 0.75) {$#2$};
    \node [style=right label] (2) at (0.25, -0.5) {$#1$};
  \end{pgfonlayer}
  \begin{pgfonlayer}{edgelayer}
    \draw [\commandkey{edgestyle}] (0.center) to (1);
  \end{pgfonlayer}
\end{tikzpicture}\,}

\newcommand{\bistateadj}[1]{\,\begin{tikzpicture}[yshift=-3mm]
  \begin{pgfonlayer}{nodelayer}
    \node [style=kpointadj, minimum width=1 cm, inner sep=2pt] (0) at (0, 0.5) {$\psi$};
    \node [style=none] (1) at (-0.75, 0.25) {};
    \node [style=none] (2) at (0.75, 0.25) {};
    \node [style=none] (3) at (-0.75, -0.5) {};
    \node [style=none] (4) at (0.75, -0.5) {};
  \end{pgfonlayer}
  \begin{pgfonlayer}{edgelayer}
    \draw (1.center) to (3.center);
    \draw (2.center) to (4.center);
  \end{pgfonlayer}
\end{tikzpicture}\,}

\newkeycommand{\pointketbra}[style1=copoint,style2=point][2]{\,%
\begin{tikzpicture}
  \begin{pgfonlayer}{nodelayer}
    \node [style=none] (0) at (0, -1.5) {};
    \node [style=\commandkey{style1}] (1) at (0, -0.75) {$#1$};
    \node [style=\commandkey{style2}] (2) at (0, 0.75) {$#2$};
    \node [style=none] (3) at (0, 1.5) {};
  \end{pgfonlayer}
  \begin{pgfonlayer}{edgelayer}
    \draw (0.center) to (1);
    \draw (2) to (3.center);
  \end{pgfonlayer}
\end{tikzpicture}\,}

\newkeycommand{\twocopointketbra}[style1=copoint,style2=copoint,style3=point][3]{\,%
\begin{tikzpicture}
  \begin{pgfonlayer}{nodelayer}
    \node [style=none] (0) at (-0.75, -1.5) {};
    \node [style=none] (1) at (0.75, -1.5) {};
    \node [style=\commandkey{style1}] (2) at (-0.75, -0.75) {$#1$};
    \node [style=\commandkey{style2}] (3) at (0.75, -0.75) {$#2$};
    \node [style=\commandkey{style3}] (4) at (0, 0.75) {$#3$};
    \node [style=none] (5) at (0, 1.5) {};
  \end{pgfonlayer}
  \begin{pgfonlayer}{edgelayer}
    \draw (0.center) to (2);
    \draw (1.center) to (3);
    \draw (4) to (5.center);
  \end{pgfonlayer}
\end{tikzpicture}\,}

\newkeycommand{\twopointketbra}[style1=copoint,style2=point,style3=point][3]{\,%
\begin{tikzpicture}
  \begin{pgfonlayer}{nodelayer}
    \node [style=none] (0) at (0, -1.5) {};
    \node [style=none] (1) at (-0.75, 1.5) {};
    \node [style=\commandkey{style1}] (2) at (0, -0.75) {$#1$};
    \node [style=\commandkey{style2}] (3) at (-0.75, 0.75) {$#2$};
    \node [style=\commandkey{style3}] (4) at (0.75, 0.75) {$#3$};
    \node [style=none] (5) at (0.75, 1.5) {};
  \end{pgfonlayer}
  \begin{pgfonlayer}{edgelayer}
    \draw (0.center) to (2);
    \draw (1.center) to (3);
    \draw (4) to (5.center);
  \end{pgfonlayer}
\end{tikzpicture}\,}

\newkeycommand{\pointbraket}[style1=point,style2=copoint][2]{\,%
\begin{tikzpicture}
  \begin{pgfonlayer}{nodelayer}
    \node [style=\commandkey{style1}] (0) at (0, -0.5) {$#1$};
    \node [style=\commandkey{style2}] (1) at (0, 0.5) {$#2$};
  \end{pgfonlayer}
  \begin{pgfonlayer}{edgelayer}
    \draw (0) to (1);
  \end{pgfonlayer}
\end{tikzpicture}\,}

\newkeycommand{\kpointbraket}[style1=kpoint,style2=kpoint adjoint][2]{\,%
\begin{tikzpicture}
  \begin{pgfonlayer}{nodelayer}
    \node [style=\commandkey{style1}] (0) at (0, -0.75) {$#1$};
    \node [style=\commandkey{style2}] (1) at (0, 0.75) {$#2$};
  \end{pgfonlayer}
  \begin{pgfonlayer}{edgelayer}
    \draw (0) to (1);
  \end{pgfonlayer}
\end{tikzpicture}\,}

\newkeycommand{\kpointmap}[style1=kpoint,style2=map][2]{\,%
\begin{tikzpicture}
  \begin{pgfonlayer}{nodelayer}
    \node [style=\commandkey{style1}] (0) at (0, -1.1) {$#1$};
    \node [style=\commandkey{style2}] (1) at (0, 0.2) {$#2$};
    \node [style=none] (2) at (0, 1.5) {};
  \end{pgfonlayer}
  \begin{pgfonlayer}{edgelayer}
    \draw (0) to (1);
    \draw (1) to (2.center);
  \end{pgfonlayer}
\end{tikzpicture}%
\,}

\newkeycommand{\sandwichtwo}[style1=point,style2=map,style3=map,style4=copoint][4]{\,%
\begin{tikzpicture}
  \begin{pgfonlayer}{nodelayer}
    \node [style=\commandkey{style2}] (0) at (0, -0.75) {$#2$};
    \node [style=\commandkey{style3}] (1) at (0, 0.75) {$#3$};
    \node [style=\commandkey{style1}] (2) at (0, -2) {$#1$};
    \node [style=\commandkey{style4}] (3) at (0, 2) {$#4$};
  \end{pgfonlayer}
  \begin{pgfonlayer}{edgelayer}
    \draw (2) to (0);
    \draw (0) to (1);
    \draw (1) to (3);
  \end{pgfonlayer}
\end{tikzpicture}\,}

\newkeycommand{\longbraket}[style1=point,style2=copoint][2]{\,%
\begin{tikzpicture}
  \begin{pgfonlayer}{nodelayer}
    \node [style=\commandkey{style1}] (0) at (0, -2) {$#1$};
    \node [style=\commandkey{style2}] (1) at (0, 2) {$#2$};
  \end{pgfonlayer}
  \begin{pgfonlayer}{edgelayer}
    \draw (0) to (1);
  \end{pgfonlayer}
\end{tikzpicture}\,}

\newkeycommand{\onb}[style=point]{\ensuremath{\left\{\tikz{\node[style=\commandkey{style}] (x) at (0,-0.3) {$j$};\draw (x)--(0,0.7);}\right\}}\xspace}

\newkeycommand{\copointmap}[style=copoint][1]{\,\tikz{\node[style=\commandkey{style}] (x) at (0,0.3) {$#1$};\draw (x)--(0,-0.7);}\,}

\tikzstyle{cdiag}=[matrix of math nodes, row sep=3em, column sep=3em, text height=1.5ex, text depth=0.25ex,inner sep=0.5em]
\tikzstyle{arrow above}=[transform canvas={yshift=0.5ex}]
\tikzstyle{arrow below}=[transform canvas={yshift=-0.5ex}]



%% file: tikzstyles.tex
\tikzstyle{env}=[copoint,regular polygon rotate=0,minimum width=0.2cm, fill=black]

\tikzstyle{every picture}=[baseline=-0.25em,scale=0.5]
\tikzstyle{dotpic}=[] 
\tikzstyle{diredges}=[every to/.style={diredge}]
\tikzstyle{math matrix}=[matrix of math nodes,left delimiter=(,right delimiter=),inner sep=2pt,column sep=1em,row sep=0.5em,nodes={inner sep=0pt},text height=1.5ex, text depth=0.25ex]

\tikzstyle{inline text}=[text height=1.5ex, text depth=0.25ex,yshift=0.5mm]
\tikzstyle{label}=[font=\footnotesize,text height=1.5ex, text depth=0.25ex,yshift=0.5mm]
\tikzstyle{left label}=[label,anchor=east,xshift=1.5mm]
\tikzstyle{right label}=[label,anchor=west,xshift=-1.5mm]

\tikzstyle{braceedge}=[decorate,decoration={brace,amplitude=2mm,raise=-1mm}]
\tikzstyle{small braceedge}=[decorate,decoration={brace,amplitude=1mm,raise=-1mm}]

\tikzstyle{doubled}=[line width=2pt] 
\tikzstyle{boldedge}=[doubled,shorten <=-0.17mm,shorten >=-0.17mm]

\tikzstyle{boldedgedashed}=[very thick,dashed,shorten <=-0.17mm,shorten >=-0.17mm]
\tikzstyle{vboldedgedashed}=[doubled,dashed,shorten <=-0.17mm,shorten >=-0.17mm]
\tikzstyle{left hook arrow}=[left hook-latex]
\tikzstyle{right hook arrow}=[right hook-latex]
\tikzstyle{sembracket}=[line width=0.5pt,shorten <=-0.07mm,shorten >=-0.07mm]

\tikzstyle{causal edge}=[->,thick,gray]
\tikzstyle{timeline}=[thick,gray, dashed]

\tikzstyle{cedge}=[<->,thick,gray!70!white]

\tikzstyle{empty diagram}=[draw=gray!50!white,dashed,shape=rectangle,minimum width=1cm,minimum height=1cm]
\tikzstyle{empty diagram small}=[draw=gray!50!white,dashed,shape=rectangle,minimum width=0.6cm,minimum height=0.5cm]

\tikzstyle{dot}=[inner sep=0.7mm,minimum width=0pt,minimum height=0pt,draw,shape=circle]
\tikzstyle{ddot}=[inner sep=0.7mm,doubled, minimum width=0pt,minimum height=0pt,draw,shape=circle]

\tikzstyle{black dot}=[dot,fill=black]
\tikzstyle{white dot}=[dot,fill=white]
\tikzstyle{green dot}=[white dot] 
\tikzstyle{gray dot}=[dot,fill=gray!40!white]
\tikzstyle{red dot}=[gray dot] 

\tikzstyle{black ddot}=[ddot,fill=black]
\tikzstyle{white ddot}=[ddot,fill=white]
\tikzstyle{gray ddot}=[ddot,fill=gray!40!white]

\tikzstyle{small dot}=[inner sep=0.4mm,minimum width=0pt,minimum height=0pt,draw,shape=circle]

\tikzstyle{small black dot}=[small dot,fill=black]
\tikzstyle{small white dot}=[small dot,fill=white]
\tikzstyle{small gray dot}=[small dot,fill=gray!40!white]

\tikzstyle{causal dot}=[inner sep=0.4mm,minimum width=0pt,minimum height=0pt,draw=white,shape=circle,fill=gray!40!white]

\tikzstyle{white phase dot}=[dot,fill=white]
\tikzstyle{white phase ddot}=[ddot,fill=white]

\tikzstyle{grey phase dot}=[dot,fill=gray!40!white]
\tikzstyle{grey phase ddot}=[ddot,fill=gray!40!white]

\tikzstyle{cnot}=[fill=white,shape=circle,inner sep=-1.4pt]
\tikzstyle{hadamard}=[square box,inner sep=0 pt,font=\tiny,minimum height=3mm,minimum width=3mm]
\tikzstyle{antipode}=[white dot,inner sep=0.3mm,font=\footnotesize]

\tikzstyle{scalar}=[diamond,draw,inner sep=0.5pt,font=\small]
\tikzstyle{dscalar}=[diamond,doubled, draw,inner sep=0.5pt,font=\small]

\tikzstyle{small box}=[rectangle,inline text,fill=white,draw,minimum height=5mm,yshift=-0.5mm,minimum width=5mm,font=\small]
\tikzstyle{small gray box}=[small box,fill=gray!30]
\tikzstyle{medium box}=[rectangle,inline text,fill=white,draw,minimum height=5mm,yshift=-0.5mm,minimum width=10mm,font=\small]
\tikzstyle{medium box bold}=[rectangle,doubled,inline text,fill=white,draw,minimum height=5mm,yshift=-0.5mm,minimum width=10mm,font=\small]
\tikzstyle{square box}=[small box] 
\tikzstyle{medium gray box}=[small box,fill=gray!30]
\tikzstyle{large box}=[rectangle,inline text,fill=white,draw,minimum height=5mm,yshift=-0.5mm,minimum width=15mm,font=\small]
\tikzstyle{large gray box}=[small box,fill=gray!30]

\tikzstyle{point}=[regular polygon,regular polygon sides=3,draw,scale=0.75,inner sep=-0.5pt,minimum width=9mm,fill=white,regular polygon rotate=180]
\tikzstyle{copoint}=[regular polygon,regular polygon sides=3,draw,scale=0.75,inner sep=-0.5pt,minimum width=9mm,fill=white]
\tikzstyle{dpoint}=[regular polygon,doubled, regular polygon sides=3,draw,scale=0.75,inner sep=-0.5pt,minimum width=9mm,fill=white,regular polygon rotate=180]
\tikzstyle{dcopoint}=[regular polygon,doubled,regular polygon sides=3,draw,scale=0.75,inner sep=-0.5pt,minimum width=9mm,fill=white]

\tikzstyle{tinypoint}=[regular polygon,regular polygon sides=3,draw,scale=0.55,inner sep=-0.15pt,minimum width=6mm,fill=white,regular polygon rotate=180]

\tikzstyle{white point}=[point]
\tikzstyle{green point}=[white point] 
\tikzstyle{white copoint}=[copoint]
\tikzstyle{gray point}=[point,fill=gray!40!white]
\tikzstyle{red point}=[gray point] 
\tikzstyle{gray copoint}=[copoint,fill=gray!40!white]

\tikzstyle{tiny gray point}=[tinypoint,fill=gray!40!white]

\tikzstyle{diredge}=[->]
\tikzstyle{rdiredge}=[<-]
\tikzstyle{thickdiredge}=[->, very thick]
\tikzstyle{pointer edge}=[->,very thick,gray]
\tikzstyle{dashed edge}=[dashed]
\tikzstyle{thick dashed edge}=[very thick,dashed]
\tikzstyle{thick gray dashed edge}=[thick dashed edge,gray!40]
\tikzstyle{thick map edge}=[very thick,|->]

\makeatletter
\newcommand{\boxshape}[3]{%
\pgfdeclareshape{#1}{
\inheritsavedanchors[from=rectangle] 
\inheritanchorborder[from=rectangle]
\inheritanchor[from=rectangle]{center}
\inheritanchor[from=rectangle]{north}
\inheritanchor[from=rectangle]{south}
\inheritanchor[from=rectangle]{west}
\inheritanchor[from=rectangle]{east}
\backgroundpath{
\southwest \pgf@xa=\pgf@x \pgf@ya=\pgf@y
\northeast \pgf@xb=\pgf@x \pgf@yb=\pgf@y

\@tempdima=#2
\@tempdimb=#3

\pgfpathmoveto{\pgfpoint{\pgf@xa - 5pt + \@tempdima}{\pgf@ya}}
\pgfpathlineto{\pgfpoint{\pgf@xa - 5pt - \@tempdima}{\pgf@yb}}
\pgfpathlineto{\pgfpoint{\pgf@xb + 5pt + \@tempdimb}{\pgf@yb}}
\pgfpathlineto{\pgfpoint{\pgf@xb + 5pt - \@tempdimb}{\pgf@ya}}
\pgfpathlineto{\pgfpoint{\pgf@xa - 5pt + \@tempdima}{\pgf@ya}}
\pgfpathclose
}
}}

\boxshape{NEbox}{0pt}{5pt}
\boxshape{SEbox}{0pt}{-5pt}
\boxshape{NWbox}{5pt}{0pt}
\boxshape{SWbox}{-5pt}{0pt}
\boxshape{EBox}{-3pt}{3pt}
\boxshape{WBox}{3pt}{-3pt}
\makeatother

\tikzstyle{map}=[draw,shape=NEbox,inner sep=2pt,minimum height=6mm,fill=white]
\tikzstyle{mapdag}=[draw,shape=SEbox,inner sep=2pt,minimum height=6mm,fill=white]
\tikzstyle{mapadj}=[draw,shape=SEbox,inner sep=2pt,minimum height=6mm,fill=white]
\tikzstyle{maptrans}=[draw,shape=SWbox,inner sep=2pt,minimum height=6mm,fill=white]
\tikzstyle{mapconj}=[draw,shape=NWbox,inner sep=2pt,minimum height=6mm,fill=white]

\tikzstyle{dmap}=[draw,doubled,shape=NEbox,inner sep=2pt,minimum height=6mm,fill=white]
\tikzstyle{dmapdag}=[draw,doubled,shape=SEbox,inner sep=2pt,minimum height=6mm,fill=white]
\tikzstyle{dmapadj}=[draw,doubled,shape=SEbox,inner sep=2pt,minimum height=6mm,fill=white]
\tikzstyle{dmaptrans}=[draw,doubled,shape=SWbox,inner sep=2pt,minimum height=6mm,fill=white]
\tikzstyle{dmapconj}=[draw,doubled,shape=NWbox,inner sep=2pt,minimum height=6mm,fill=white]

\tikzstyle{ddmap}=[draw,doubled,dashed,shape=NEbox,inner sep=2pt,minimum height=6mm,fill=white]
\tikzstyle{ddmapdag}=[draw,doubled,dashed,shape=SEbox,inner sep=2pt,minimum height=6mm,fill=white]
\tikzstyle{ddmapadj}=[draw,doubled,dashed,shape=SEbox,inner sep=2pt,minimum height=6mm,fill=white]
\tikzstyle{ddmaptrans}=[draw,doubled,dashed,shape=SWbox,inner sep=2pt,minimum height=6mm,fill=white]
\tikzstyle{ddmapconj}=[draw,doubled,dashed,shape=NWbox,inner sep=2pt,minimum height=6mm,fill=white]

\boxshape{sNEbox}{0pt}{3pt}
\boxshape{sSEbox}{0pt}{-3pt}
\boxshape{sNWbox}{3pt}{0pt}
\boxshape{sSWbox}{-3pt}{0pt}
\tikzstyle{smap}=[draw,shape=sNEbox,fill=white]
\tikzstyle{smapdag}=[draw,shape=sSEbox,fill=white]
\tikzstyle{smapadj}=[draw,shape=sSEbox,fill=white]
\tikzstyle{smaptrans}=[draw,shape=sSWbox,fill=white]
\tikzstyle{smapconj}=[draw,shape=sNWbox,fill=white]

\tikzstyle{dsmap}=[draw,dashed,shape=sNEbox,fill=white]
\tikzstyle{dsmapdag}=[draw,dashed,shape=sSEbox,fill=white]
\tikzstyle{dsmaptrans}=[draw,dashed,shape=sSWbox,fill=white]
\tikzstyle{dsmapconj}=[draw,dashed,shape=sNWbox,fill=white]

\boxshape{mNEbox}{0pt}{10pt}
\boxshape{mSEbox}{0pt}{-10pt}
\boxshape{mNWbox}{10pt}{0pt}
\boxshape{mSWbox}{-10pt}{0pt}
\tikzstyle{mmap}=[draw,shape=mNEbox]
\tikzstyle{mmapdag}=[draw,shape=mSEbox]
\tikzstyle{mmaptrans}=[draw,shape=mSWbox]
\tikzstyle{mmapconj}=[draw,shape=mNWbox]

\tikzstyle{mmapgray}=[draw,fill=gray!40!white,shape=mNEbox]
\tikzstyle{smapgray}=[draw,fill=gray!40!white,shape=sNEbox]

\makeatletter
\pgfdeclareshape{cornerpoint}{
\inheritsavedanchors[from=rectangle] 
\inheritanchorborder[from=rectangle]
\inheritanchor[from=rectangle]{center}
\inheritanchor[from=rectangle]{north}
\inheritanchor[from=rectangle]{south}
\inheritanchor[from=rectangle]{west}
\inheritanchor[from=rectangle]{east}
\backgroundpath{
\southwest \pgf@xa=\pgf@x \pgf@ya=\pgf@y
\northeast \pgf@xb=\pgf@x \pgf@yb=\pgf@y

\pgfmathsetmacro{\pgf@shorten@left}{\pgfkeysvalueof{/tikz/shorten left}}
\pgfmathsetmacro{\pgf@shorten@right}{\pgfkeysvalueof{/tikz/shorten right}}

\pgfpathmoveto{\pgfpoint{0.5 * (\pgf@xa + \pgf@xb)}{\pgf@ya - 5pt}}
\pgfpathlineto{\pgfpoint{\pgf@xa - 8pt + \pgf@shorten@left}{\pgf@yb - 1.5 * \pgf@shorten@left}}
\pgfpathlineto{\pgfpoint{\pgf@xa - 8pt + \pgf@shorten@left}{\pgf@yb}}
\pgfpathlineto{\pgfpoint{\pgf@xb + 8pt - \pgf@shorten@right}{\pgf@yb}}
\pgfpathlineto{\pgfpoint{\pgf@xb + 8pt - \pgf@shorten@right}{\pgf@yb - 1.5 * \pgf@shorten@right}}
\pgfpathclose
}
}

\pgfdeclareshape{cornercopoint}{
\inheritsavedanchors[from=rectangle] 
\inheritanchorborder[from=rectangle]
\inheritanchor[from=rectangle]{center}
\inheritanchor[from=rectangle]{north}
\inheritanchor[from=rectangle]{south}
\inheritanchor[from=rectangle]{west}
\inheritanchor[from=rectangle]{east}
\backgroundpath{
\southwest \pgf@xa=\pgf@x \pgf@ya=\pgf@y
\northeast \pgf@xb=\pgf@x \pgf@yb=\pgf@y

\pgfmathsetmacro{\pgf@shorten@left}{\pgfkeysvalueof{/tikz/shorten left}}
\pgfmathsetmacro{\pgf@shorten@right}{\pgfkeysvalueof{/tikz/shorten right}}

\pgfpathmoveto{\pgfpoint{0.5 * (\pgf@xa + \pgf@xb)}{\pgf@yb + 5pt}}
\pgfpathlineto{\pgfpoint{\pgf@xa - 8pt + \pgf@shorten@left}{\pgf@ya + 1.5 * \pgf@shorten@left}}
\pgfpathlineto{\pgfpoint{\pgf@xa - 8pt + \pgf@shorten@left}{\pgf@ya}}
\pgfpathlineto{\pgfpoint{\pgf@xb + 8pt - \pgf@shorten@right}{\pgf@ya}}
\pgfpathlineto{\pgfpoint{\pgf@xb + 8pt - \pgf@shorten@right}{\pgf@ya + 1.5 * \pgf@shorten@right}}
\pgfpathclose
}
}

\makeatother

\pgfkeyssetvalue{/tikz/shorten left}{0pt}
\pgfkeyssetvalue{/tikz/shorten right}{0pt}

\tikzstyle{kpoint common}=[draw,fill=white,inner sep=1pt,minimum height=4mm]
\tikzstyle{kpoint}=[shape=cornerpoint,shorten left=5pt,kpoint common]
\tikzstyle{kpoint adjoint}=[shape=cornercopoint,shorten left=5pt,kpoint common]
\tikzstyle{kpoint conjugate}=[shape=cornerpoint,shorten right=5pt,kpoint common]
\tikzstyle{kpoint transpose}=[shape=cornercopoint,shorten right=5pt,kpoint common]
\tikzstyle{kpoint symm}=[shape=cornerpoint,shorten left=5pt,shorten right=5pt,kpoint common]

\tikzstyle{kpointdag}=[kpoint adjoint]
\tikzstyle{kpointadj}=[kpoint adjoint]
\tikzstyle{kpointconj}=[kpoint conjugate]
\tikzstyle{kpointtrans}=[kpoint transpose]

\tikzstyle{dkpoint}=[kpoint,doubled]
\tikzstyle{dkpointdag}=[kpoint adjoint,doubled]
\tikzstyle{dkpointadj}=[kpoint adjoint,doubled]
\tikzstyle{dkpointconj}=[kpoint conjugate,doubled]
\tikzstyle{dkpointtrans}=[kpoint transpose,doubled]

\tikzstyle{kscalar}=[kpoint common, shape=EBox, inner xsep=-1pt, inner ysep=3pt,font=\small]
\tikzstyle{kscalarconj}=[kpoint common, shape=WBox, inner xsep=-1pt, inner ysep=3pt,font=\small]

 \tikzstyle{upground}=[circuit ee IEC,thick,ground,rotate=90,scale=2.5]
 \tikzstyle{downground}=[circuit ee IEC,thick,ground,rotate=-90,scale=2.5]
 \tikzstyle{bigground}=[regular polygon,regular polygon sides=3,draw=gray,scale=0.50,inner sep=-0.5pt,minimum width=10mm,fill=gray]

\tikzstyle{arrs}=[-latex,font=\small,auto]
\tikzstyle{arrow plain}=[arrs]
\tikzstyle{arrow dashed}=[dashed,arrs]
\tikzstyle{arrow bold}=[very thick,arrs]
\tikzstyle{arrow hide}=[draw=white!0,-]
\tikzstyle{arrow reverse}=[latex-]
\tikzstyle{cdnode}=[]

%% file: compound-process.tikz
\begin{tikzpicture}
	\begin{pgfonlayer}{nodelayer}
		\node [style={medium box}] (0) at (0, 1) {$g$};
		\node [style=none] (1) at (-1.25, -0.75) {};
		\node [style=none] (2) at (-0.75, 0.5) {};
		\node [style=none] (3) at (-0.75, 2.5) {};
		\node [style=none] (4) at (-0.75, 1.5) {};
		\node [style={right label}] (5) at (-2, 2.25) {$A$};  
		\node [style=none] (6) at (0.75, 0.5) {};
		\node [style=none] (7) at (1.25, -0.75) {};
		\node [style={medium box}] (8) at (-1.75, -1.25) {$f$};
		\node [style=none] (9) at (-2.25, -0.75) {};
		\node [style=none] (10) at (-2.25, 2.5) {};
		\node [style=none] (11) at (1.25, -2.5) {};
		\node [style=none] (12) at (1.25, -1.75) {};
		\node [style={small box}] (13) at (1.25, -1.25) {$h$};
		\node [style=none] (14) at (0.75, 1.5) {};
		\node [style=none] (15) at (0.75, 2.5) {};
		\node [style={right label}] (16) at (-0.5, 2.25) {$B$};
		\node [style={right label}] (17) at (1, 2.25) {$C$};
		\node [style={right label}] (18) at (-0.75, -0.25) {$A$};
		\node [style={right label}] (19) at (1.25, 0) {$D$};
		\node [style={right label}] (20) at (1.5, -2.25) {$A$};
	\end{pgfonlayer}
	\begin{pgfonlayer}{edgelayer}
		\draw [in=-90, out=90, looseness=1.00] (1.center) to (2.center);
		\draw (4.center) to (3.center);
		\draw [in=-90, out=90, looseness=1.00] (7.center) to (6.center);
		\draw (9.center) to (10.center);
		\draw (11.center) to (12.center);
		\draw (14.center) to (15.center);
	\end{pgfonlayer}
\end{tikzpicture}

%% file: compound-process-mw.tikz
\begin{tikzpicture}
	\begin{pgfonlayer}{nodelayer}
		\node [style=medium box] (0) at (0, 1) {$g$};
		\node [style=none] (1) at (-1.25, -0.75) {};
		\node [style=none] (2) at (-0.75, 0.5) {};
		\node [style=none] (3) at (-0.75, 2.5) {};
		\node [style=none] (4) at (-0.75, 1.5) {};
		\node [style=none] (5) at (0.75, 0.5) {};
		\node [style=none] (6) at (1.25, -0.75) {};
		\node [style=medium box] (7) at (-1.75, -1.25) {$f$};
		\node [style=none] (8) at (-2.25, -0.75) {};
		\node [style=none] (9) at (-2.25, 2.5) {};
		\node [style=none] (10) at (1.25, -2.5) {};
		\node [style=none] (11) at (1.25, -1.75) {};
		\node [style=small box] (12) at (1.25, -1.25) {$h$};
		\node [style=none] (13) at (0.75, 1.5) {};
		\node [style=none] (14) at (0.75, 2.5) {};
	\end{pgfonlayer}
	\begin{pgfonlayer}{edgelayer}
		\draw [in=-90, out=90, looseness=1.00] (1.center) to (2.center);
		\draw [style=dashed] (4.center) to (3.center);
		\draw [style=boldedgedashed, in=-90, out=90, looseness=1.00] (6.center) to (5.center);
		\draw (8.center) to (9.center);
		\draw (10.center) to (11.center);
		\draw [style=boldedge] (13.center) to (14.center);
	\end{pgfonlayer}
\end{tikzpicture}

%% file: disc.tikz
\begin{tikzpicture}
	\begin{pgfonlayer}{nodelayer}
		\node [style=none] (0) at (0, -0.75) {};
		\node [style=upground] (1) at (0, 0.75) {};
	\end{pgfonlayer}
	\begin{pgfonlayer}{edgelayer}
		\draw [style=swap] (1) to (0.center); 
	\end{pgfonlayer}
\end{tikzpicture}

%% file: norm.tikz
\begin{tikzpicture}
	\begin{pgfonlayer}{nodelayer}
		\node [style=none] (0) at (-1.75, 0.5) {};
		\node [style=none] (1) at (-1.75, -1.25) {};
		\node [style=none] (2) at (-1.75, -0.5) {};
		\node [style=small box] (3) at (-1.75, 0) {$f$};
		\node [style=upground] (4) at (-1.75, 1.25) {};
		\node [style=none] (5) at (1.75, -1.25) {};
		\node [style=upground] (6) at (1.75, -0.25) {};
		\node [style=none] (7) at (0, 0) {$=$};
	\end{pgfonlayer}
	\begin{pgfonlayer}{edgelayer}
		\draw (1.center) to (2.center);
		\draw [style=swap] (4) to (0.center);
		\draw [style=swap] (6) to (5.center);
	\end{pgfonlayer}
\end{tikzpicture}

%% file: nosig1.tikz
\begin{tikzpicture}
	\begin{pgfonlayer}{nodelayer}
		\node [style=none] (0) at (-3, 0.5) {};
		\node [style=none] (1) at (-3, -1.25) {};
		\node [style=none] (2) at (-3, -0.5) {};
		\node [style=upground] (3) at (-3, 1.25) {};
		\node [style=none] (4) at (1.75, -1.25) {};
		\node [style=upground] (5) at (1.75, -0.25) {};
		\node [style=none] (6) at (0, 0) {$=$};
		\node [style=medium box] (7) at (-2.25, 0) {$f$};
		\node [style=none] (8) at (-1.5, -1.25) {};
		\node [style=none] (9) at (-1.5, -0.5) {};
		\node [style=none] (10) at (-1.5, 0.5) {};
		\node [style=none] (11) at (-1.5, 1.25) {};
		\node [style=small box] (12) at (3.25, 0) {$h$};
		\node [style=none] (13) at (3.25, -1.25) {};
		\node [style=none] (14) at (3.25, -0.5) {};
		\node [style=none] (15) at (3.25, 0.5) {};
		\node [style=none] (16) at (3.25, 1.25) {};
	\end{pgfonlayer}
	\begin{pgfonlayer}{edgelayer}
		\draw (1.center) to (2.center);
		\draw [style=swap] (3) to (0.center);
		\draw [style=swap] (5) to (4.center);
		\draw (8.center) to (9.center);
		\draw (10.center) to (11.center);
		\draw (13.center) to (14.center);
		\draw (15.center) to (16.center);
	\end{pgfonlayer}
\end{tikzpicture}

%% file: nosig2.tikz
\begin{tikzpicture}
	\begin{pgfonlayer}{nodelayer}
		\node [style=none] (0) at (-1.5, 0.5) {};
		\node [style=medium box] (1) at (-2.25, 0) {$f$};
		\node [style=none] (2) at (1.75, 0.5) {};
		\node [style=none] (3) at (-3, -0.5) {};
		\node [style=upground] (4) at (3.25, -0.25) {};
		\node [style=none] (5) at (-3, 0.5) {};
		\node [style=none] (6) at (1.75, -0.5) {};
		\node [style=none] (7) at (3.25, -1.25) {};
		\node [style=none] (8) at (-3, -1.25) {};
		\node [style=upground] (9) at (-1.5, 1.25) {};
		\node [style=none] (10) at (-1.5, -0.5) {};
		\node [style=none] (11) at (1.75, -1.25) {};
		\node [style=small box] (12) at (1.75, 0) {$h'$};
		\node [style=none] (13) at (-1.5, -1.25) {};
		\node [style=none] (14) at (0, 0) {$=$};
		\node [style=none] (15) at (1.75, 1.25) {};
		\node [style=none] (16) at (-3, 1.25) {};
	\end{pgfonlayer}
	\begin{pgfonlayer}{edgelayer}
		\draw (8.center) to (3.center);
		\draw [style=swap] (9) to (0.center);
		\draw [style=swap] (4) to (7.center);
		\draw (13.center) to (10.center);
		\draw (5.center) to (16.center);
		\draw (11.center) to (6.center);
		\draw (2.center) to (15.center);
	\end{pgfonlayer}
\end{tikzpicture}

%% file: swap.tikz
\begin{tikzpicture}
	\begin{pgfonlayer}{nodelayer}
		\node [style=none] (0) at (0.75, -0.75) {};
		\node [style=none] (1) at (-0.75, 0.75) {};
		\node [style=none] (2) at (0.75, 0.75) {};
		\node [style=none] (3) at (-0.75, -0.75) {};
	\end{pgfonlayer}
	\begin{pgfonlayer}{edgelayer}
		\draw [style=swap, in=90, out=-90, looseness=1.00] (1) to (0.center);
		\draw [style=swap, in=90, out=-90, looseness=1.00] (2.center) to (3.center);
	\end{pgfonlayer}
\end{tikzpicture}

%% file: nosig3.tikz
\begin{tikzpicture}
	\begin{pgfonlayer}{nodelayer}
		\node [style=none] (0) at (-3, 0.5) {};
		\node [style=none] (1) at (-3, -1.25) {};
		\node [style=none] (2) at (-3, -0.5) {};
		\node [style=upground] (3) at (-3, 1.25) {};
		\node [style=none] (4) at (1.75, -1.25) {};
		\node [style=copoint] (5) at (1.75, -0.25) {$e$};
		\node [style=none] (6) at (0, 0) {$=$};
		\node [style=medium box] (7) at (-2.25, 0) {$f$};
		\node [style=none] (8) at (-1.5, -1.25) {};
		\node [style=none] (9) at (-1.5, -0.5) {};
		\node [style=none] (10) at (-1.5, 0.5) {};
		\node [style=none] (11) at (-1.5, 1.25) {};
		\node [style=small box] (12) at (3.25, 0) {$h$};
		\node [style=none] (13) at (3.25, -1.25) {};
		\node [style=none] (14) at (3.25, -0.5) {};
		\node [style=none] (15) at (3.25, 0.5) {};
		\node [style=none] (16) at (3.25, 1.25) {};
	\end{pgfonlayer}
	\begin{pgfonlayer}{edgelayer}
		\draw (1.center) to (2.center);
		\draw [style=swap] (3) to (0.center);
		\draw [style=swap] (5) to (4.center);
		\draw (8.center) to (9.center);
		\draw (10.center) to (11.center);
		\draw (13.center) to (14.center);
		\draw (15.center) to (16.center);
	\end{pgfonlayer}
\end{tikzpicture}

%% file: nosig4.tikz
\begin{tikzpicture}
	\begin{pgfonlayer}{nodelayer}
		\node [style=none] (0) at (-1.5, 0.5) {};
		\node [style=medium box] (1) at (-2.25, 0) {$f$};
		\node [style=none] (2) at (1.75, 0.5) {};
		\node [style=none] (3) at (-3, -0.5) {};
		\node [style=copoint] (4) at (3.25, -0.25) {$e'$};
		\node [style=none] (5) at (-3, 0.5) {};
		\node [style=none] (6) at (1.75, -0.5) {};
		\node [style=none] (7) at (3.25, -1.25) {};
		\node [style=none] (8) at (-3, -1.25) {};
		\node [style=upground] (9) at (-1.5, 1.25) {};
		\node [style=none] (10) at (-1.5, -0.5) {};
		\node [style=none] (11) at (1.75, -1.25) {};
		\node [style=small box] (12) at (1.75, 0) {$h'$};
		\node [style=none] (13) at (-1.5, -1.25) {};
		\node [style=none] (14) at (0, 0) {$=$};
		\node [style=none] (15) at (1.75, 1.25) {};
		\node [style=none] (16) at (-3, 1.25) {};
	\end{pgfonlayer}
	\begin{pgfonlayer}{edgelayer}
		\draw (8.center) to (3.center);
		\draw [style=swap] (9) to (0.center);
		\draw [style=swap] (4) to (7.center);
		\draw (13.center) to (10.center);
		\draw (5.center) to (16.center);
		\draw (11.center) to (6.center);
		\draw (2.center) to (15.center);
	\end{pgfonlayer}
\end{tikzpicture}

%% file: Causaly1.tikz
\begin{tikzpicture}
	\begin{pgfonlayer}{nodelayer}
		\node [style=causal dot] (0) at (2, 0) {};
		\node [style=causal dot] (1) at (-2, 0) {};
		\node [style=causal dot] (2) at (0, -2) {};
		\node [style=causal dot] (3) at (0, 2) {};
		\node [style=none] (4) at (-1.5, -2.5) {};
		\node [style=left label] (5) at (-1.75, -2.5) {$\bot$};
		\node [style=none] (6) at (-0.25, -2.25) {};
		\node [style=none] (7) at (0.25, 2.25) {};
		\node [style=right label] (8) at (1.75, 2.5) {$\top$};
		\node [style=none] (9) at (1.5, 2.5) {};
		\node [style=right label] (10) at (3.75, 0.5) {$b$};
		\node [style=none] (11) at (2.25, 0.25) {};
		\node [style=none] (12) at (3.5, 0.5) {};
		\node [style=left label] (13) at (-3.75, -0.5) {$a$};
		\node [style=none] (14) at (-3.5, -0.5) {};
		\node [style=none] (15) at (-2.25, -0.25) {};
	\end{pgfonlayer}
	\begin{pgfonlayer}{edgelayer}
		\draw [style=causal edge] (2) to (0);
		\draw [style=causal edge] (2) to (1);
		\draw [style=causal edge] (1) to (3);
		\draw [style=causal edge] (0) to (3);
		\draw [style=diredge, in=-150, out=0, looseness=0.75] (4.center) to (6.center);
		\draw [style=diredge, in=30, out=180, looseness=0.75] (9.center) to (7.center);
		\draw [style=diredge, in=30, out=180, looseness=0.75] (12.center) to (11.center);
		\draw [style=diredge, in=-150, out=0, looseness=0.75] (14.center) to (15.center);
	\end{pgfonlayer}
\end{tikzpicture}

%% file: diamond.tikz
\begin{tikzpicture}
	\begin{pgfonlayer}{nodelayer}
		\node [style=none] (0) at (3, -1) {};
		\node [style=none] (1) at (3, -0.5) {};
		\node [style=none] (2) at (1.5, -0.5) {};
		\node [style=none] (3) at (-0.75, -1.75) {};
		\node [style=none] (4) at (0.75, -1.75) {};
		\node [style=none] (5) at (-1.5, -0.5) {};
		\node [style=medium box] (6) at (0, -2.25) {$f_\bot$};
		\node [style=medium box] (7) at (0, 2.25) {$f_\top$};
		\node [style=none] (8) at (1.5, 0.5) {};
		\node [style=none] (9) at (0.75, 1.75) {};
		\node [style=medium box] (10) at (-2.25, 0) {$f_a$};
		\node [style=medium box] (11) at (2.25, 0) {$f_b$};
		\node [style=none] (12) at (-0.75, 1.75) {};
		\node [style=none] (13) at (-1.5, 0.5) {};
		\node [style=none] (14) at (3, 1) {};
		\node [style=none] (15) at (3, 0.5) {};
		\node [style=none] (16) at (-3, -0.5) {};
		\node [style=none] (17) at (-3, -1) {};
		\node [style=none] (18) at (-3, 1) {};
		\node [style=none] (19) at (-3, 0.5) {};
	\end{pgfonlayer}
	\begin{pgfonlayer}{edgelayer}
		\draw (0.center) to (1.center);
		\draw [in=-90, out=90, looseness=1.00] (4.center) to (2.center);
		\draw [in=-90, out=90, looseness=1.00] (3.center) to (5.center);
		\draw [in=-90, out=90, looseness=1.00] (8.center) to (9.center);
		\draw [in=-90, out=90, looseness=1.00] (13.center) to (12.center);
		\draw (15.center) to (14.center);
		\draw (17.center) to (16.center);
		\draw (19.center) to (18.center);
	\end{pgfonlayer}
\end{tikzpicture}

%% file: Causaly2.tikz
\begin{tikzpicture}
	\begin{pgfonlayer}{nodelayer}
		\node [style=causal dot] (0) at (8.5, 0) {};
		\node [style=causal dot] (1) at (4.5, 0) {};
		\node [style=causal dot] (2) at (6.5, -2) {};
		\node [style=causal dot] (3) at (6.5, 2) {};
		\node [style=causal dot] (4) at (-8.5, 0) {};
		\node [style=causal dot] (5) at (-6.5, 2) {};
		\node [style=causal dot] (6) at (-4.5, 0) {};
		\node [style=causal dot] (7) at (-6.5, -2) {};
		\node [style=causal dot] (8) at (-6.5, -3) {};
		\node [style=causal dot] (9) at (-6.5, -4) {};
		\node [style=causal dot] (10) at (-6.5, 3) {};
		\node [style=causal dot] (11) at (-6.5, 4) {};
		\node [style=causal dot] (12) at (9.5, 1) {};
		\node [style=causal dot] (13) at (9.5, -1) {};
		\node [style=causal dot] (14) at (3.5, 1) {};
		\node [style=causal dot] (15) at (3.5, -1) {};
		\node [style=none] (16) at (-9.5, 0) {};
	\end{pgfonlayer}
	\begin{pgfonlayer}{edgelayer}
		\draw [style=causal edge] (2) to (0);
		\draw [style=causal edge] (2) to (1);
		\draw [style=causal edge] (1) to (3);
		\draw [style=causal edge] (0) to (3);
		\draw [style=causal edge] (7) to (6);
		\draw [style=causal edge] (7) to (4);
		\draw [style=causal edge] (4) to (5);
		\draw [style=causal edge] (6) to (5);
		\draw [style=causal edge] (8) to (4);
		\draw [style=causal edge] (8) to (6);
		\draw [style=causal edge] (8) to (7);
		\draw [style=causal edge] (9) to (6);
		\draw [style=causal edge] (9) to (4);
		\draw [style=causal edge] (10) to (11);
		\draw [style=causal edge] (6) to (10);
		\draw [style=causal edge] (6) to (11);
		\draw [style=causal edge] (4) to (11);
		\draw [style=causal edge] (4) to (10);
		\draw [style=causal edge] (1) to (14);
		\draw [style=causal edge] (15) to (1);
		\draw [style=causal edge] (13) to (0);
		\draw [style=causal edge] (0) to (12);
	\end{pgfonlayer}
\end{tikzpicture}

%% file: Causaly3.tikz
\begin{tikzpicture}
	\begin{pgfonlayer}{nodelayer}
		\node [style=causal dot] (0) at (8.5, 0) {};
		\node [style=causal dot] (1) at (4.5, 0) {};
		\node [style=causal dot] (2) at (6.5, -2) {};
		\node [style=causal dot] (3) at (6.5, 2) {};
		\node [style=causal dot] (4) at (-8.5, 0) {};
		\node [style=causal dot] (5) at (-6.5, 2) {};
		\node [style=causal dot] (6) at (-4.5, 0) {};
		\node [style=causal dot] (7) at (-6.5, -2) {};
		\node [style=causal dot] (8) at (-6.5, -3) {};
		\node [style=causal dot] (9) at (-6.5, -4) {};
		\node [style=causal dot] (10) at (-6.5, 3) {};
		\node [style=causal dot] (11) at (-6.5, 4) {};
		\node [style=causal dot] (12) at (9.5, 1) {};
		\node [style=causal dot] (13) at (9.5, -1) {};
		\node [style=causal dot] (14) at (3.5, 1) {};
		\node [style=causal dot] (15) at (3.5, -1) {};
		\node [style=none] (16) at (-9.5, 0) {};
		\node [style=none] (17) at (3.5, 1.75) {};
		\node [style=none] (18) at (5.25, 0) {};
		\node [style=none] (19) at (2.75, 0) {};
		\node [style=none] (20) at (3.5, 1.75) {};
		\node [style=none] (21) at (3.5, -1.75) {};
		\node [style=none] (22) at (2.75, 0) {};
		\node [style=none] (23) at (5.25, 0) {};
		\node [style=none] (24) at (3.5, -1.75) {};
		\node [style=none] (25) at (-6.5, 1.25) {};
		\node [style=none] (26) at (-6.5, 4.75) {};
		\node [style=none] (27) at (-7.25, 3) {};
		\node [style=none] (28) at (-5.75, 3) {};
		\node [style=none] (29) at (-6.5, 1.25) {};
		\node [style=none] (30) at (-5.75, 3) {};
		\node [style=none] (31) at (-6.5, 4.75) {};
		\node [style=none] (32) at (-7.25, 3) {};
		\node [style=none] (33) at (-6.5, -4.75) {};
		\node [style=none] (34) at (-6.5, -4.75) {};
		\node [style=none] (35) at (-6.5, -1.25) {};
		\node [style=none] (36) at (-7.25, -3) {};
		\node [style=none] (37) at (-5.75, -3) {};
		\node [style=none] (38) at (-7.25, -3) {};
		\node [style=none] (39) at (-5.75, -3) {};
		\node [style=none] (40) at (-6.5, -1.25) {};
		\node [style=none] (41) at (9.5, -1.75) {};
		\node [style=none] (42) at (9.5, 1.75) {};
		\node [style=none] (43) at (10.25, 0) {};
		\node [style=none] (44) at (7.75, 0) {};
		\node [style=none] (45) at (9.5, -1.75) {};
		\node [style=none] (46) at (7.75, 0) {};
		\node [style=none] (47) at (9.5, 1.75) {};
		\node [style=none] (48) at (10.25, 0) {};
	\end{pgfonlayer}
	\begin{pgfonlayer}{edgelayer}
		\draw [style=causal edge] (2) to (0);
		\draw [style=causal edge] (2) to (1);
		\draw [style=causal edge] (1) to (3);
		\draw [style=causal edge] (0) to (3);
		\draw [style=causal edge] (7) to (6);
		\draw [style=causal edge] (7) to (4);
		\draw [style=causal edge] (4) to (5);
		\draw [style=causal edge] (6) to (5);
		\draw [style=causal edge] (8) to (4);
		\draw [style=causal edge] (8) to (6);
		\draw [style=causal edge] (8) to (7);
		\draw [style=causal edge] (9) to (6);
		\draw [style=causal edge] (9) to (4);
		\draw [style=causal edge] (10) to (11);
		\draw [style=causal edge] (6) to (10);
		\draw [style=causal edge] (6) to (11);
		\draw [style=causal edge] (4) to (11);
		\draw [style=causal edge] (4) to (10);
		\draw [style=causal edge] (1) to (14);
		\draw [style=causal edge] (15) to (1);
		\draw [style=causal edge] (13) to (0);
		\draw [style=causal edge] (0) to (12);
		\draw [style=dashed, in=90, out=0, looseness=0.75] (17.center) to (18.center);
		\draw [style=dashed, in=90, out=180, looseness=0.75] (20.center) to (19.center);
		\draw [style=dashed, in=-90, out=0, looseness=0.75] (24.center) to (23.center);
		\draw [style=dashed, in=-90, out=180, looseness=0.75] (21.center) to (22.center);
		\draw [style=dashed, in=90, out=0, looseness=0.75] (31.center) to (28.center);
		\draw [style=dashed, in=90, out=180, looseness=0.75] (26.center) to (27.center);
		\draw [style=dashed, in=-90, out=0, looseness=0.75] (25.center) to (30.center);
		\draw [style=dashed, in=-90, out=180, looseness=0.75] (29.center) to (32.center);
		\draw [style=dashed, in=90, out=0, looseness=0.75] (40.center) to (39.center);
		\draw [style=dashed, in=90, out=180, looseness=0.75] (35.center) to (38.center);
		\draw [style=dashed, in=-90, out=0, looseness=0.75] (34.center) to (37.center);
		\draw [style=dashed, in=-90, out=180, looseness=0.75] (33.center) to (36.center);
		\draw [style=dashed, in=90, out=180, looseness=0.75] (47.center) to (44.center);
		\draw [style=dashed, in=90, out=0, looseness=0.75] (42.center) to (43.center);
		\draw [style=dashed, in=-90, out=180, looseness=0.75] (41.center) to (46.center);
		\draw [style=dashed, in=-90, out=0, looseness=0.75] (45.center) to (48.center);
	\end{pgfonlayer}
\end{tikzpicture}

%% file: diamondV1.tikz
\begin{tikzpicture}
	\begin{pgfonlayer}{nodelayer}
		\node [style=none] (0) at (3, -1) {};
		\node [style=none] (1) at (3, -0.5) {};
		\node [style=none] (2) at (1.5, -0.5) {};
		\node [style=none] (3) at (-0.75, -1.75) {};
		\node [style=none] (4) at (0.75, -1.75) {};
		\node [style=none] (5) at (-1.5, -0.5) {};
		\node [style=medium box] (6) at (0, -2.25) {$h_\bot$};
		\node [style=medium box] (7) at (-2.25, 0) {$h_a$};
		\node [style=medium box] (8) at (2.25, 0) {$h_b$};
		\node [style=none] (9) at (3, 1) {};
		\node [style=none] (10) at (3, 0.5) {};
		\node [style=none] (11) at (-3, -0.5) {};
		\node [style=none] (12) at (-3, -1) {};
		\node [style=none] (13) at (-3, 1) {};
		\node [style=none] (14) at (-3, 0.5) {};
	\end{pgfonlayer}
	\begin{pgfonlayer}{edgelayer}
		\draw (0.center) to (1.center);
		\draw [in=-90, out=90, looseness=1.00] (4.center) to (2.center);
		\draw [in=-90, out=90, looseness=1.00] (3.center) to (5.center);
		\draw (10.center) to (9.center);
		\draw (12.center) to (11.center);
		\draw (14.center) to (13.center);
	\end{pgfonlayer}
\end{tikzpicture}

%% file: diamondV2.tikz
\begin{tikzpicture}
	\begin{pgfonlayer}{nodelayer}
		\node [style=none] (0) at (3, -1) {};
		\node [style=none] (1) at (3, -0.5) {};
		\node [style=none] (2) at (1.5, -0.5) {};
		\node [style=none] (3) at (-0.75, -1.75) {};
		\node [style=none] (4) at (0.75, -1.75) {};
		\node [style=none] (5) at (-1.5, -0.5) {};
		\node [style=medium box bold] (6) at (0, -2.25) {$f_\bot$};
		\node [style=medium box bold] (7) at (-2.25, 0) {$f_a$};
		\node [style=medium box bold] (8) at (2.25, 0) {$f_b$};
		\node [style=none] (9) at (3, 1) {};
		\node [style=none] (10) at (3, 0.5) {};
		\node [style=none] (11) at (-3, -0.5) {};
		\node [style=none] (12) at (-3, -1) {};
		\node [style=none] (13) at (-3, 1) {};
		\node [style=none] (14) at (-3, 0.5) {};
	\end{pgfonlayer}
	\begin{pgfonlayer}{edgelayer}
		\draw (0.center) to (1.center);
		\draw [boldedge, in=-90, out=90, looseness=1.00] (4.center) to (2.center);
		\draw [boldedge, in=-90, out=90, looseness=1.00] (3.center) to (5.center);
		\draw (10.center) to (9.center);
		\draw (12.center) to (11.center);
		\draw (14.center) to (13.center);
	\end{pgfonlayer}
\end{tikzpicture}

%% file: box.tikz
\begin{tikzpicture}
	\begin{pgfonlayer}{nodelayer}
		\node [style=none] (0) at (0.75, 1) {};
		\node [style=none] (1) at (-0.75, -0.5) {};
		\node [style=none] (2) at (-0.75, 1) {};
		\node [style=medium box] (3) at (0, 0) {$f$};
		\node [style=none] (4) at (-0.75, -1) {};
		\node [style=none] (5) at (-0.75, 0.5) {};
		\node [style=none] (6) at (0.75, 0.5) {};
		\node [style=none] (7) at (0.75, -0.5) {};
		\node [style=none] (8) at (0.75, -1) {};
	\end{pgfonlayer}
	\begin{pgfonlayer}{edgelayer}
		\draw (4.center) to (1.center);
		\draw (5.center) to (2.center);
		\draw [style=swap] (0) to (6.center);
		\draw [style=swap] (7.center) to (8.center);
	\end{pgfonlayer}
\end{tikzpicture}

%% file: box1.tikz
\begin{tikzpicture}
	\begin{pgfonlayer}{nodelayer}
		\node [style=none] (0) at (-5.75, 1.5) {};
		\node [style=none] (1) at (-7.25, -1.5) {};
		\node [style=none] (2) at (-7.25, 1.5) {};
		\node [style=none] (3) at (-5.75, -1.5) {};
		\node [style=none] (4) at (-5.75, 0.75) {};
		\node [style=small box] (5) at (-7.25, 0.5) {};
		\node [style=upground] (6) at (-5.75, -0.75) {};
		\node [style=none] (7) at (-7.25, -1.25) {};
		\node [style=none] (8) at (-5, 1.25) {};
		\node [style=none] (9) at (-8, 1.25) {};
		\node [style=none] (10) at (-8, -1.25) {};
		\node [style=none] (11) at (-5, -1.25) {};
		\node [style=none] (12) at (8, -1.25) {};
		\node [style=none] (13) at (7.25, 1.5) {};
		\node [style=none] (14) at (5, -1.25) {};
		\node [style=none] (15) at (8, 1.25) {};
		\node [style=none] (16) at (5.75, -1) {};
		\node [style=none] (17) at (5, 1.25) {};
		\node [style=none] (18) at (7.25, 1) {};
		\node [style=none] (19) at (5.75, -1.5) {};
		\node [style=none] (20) at (7.25, -1.5) {};
		\node [style=none] (21) at (5.75, 1.5) {};
		\node [style=none] (22) at (7.25, -1) {};
		\node [style=none] (23) at (5.75, 1) {};
	\end{pgfonlayer}
	\begin{pgfonlayer}{edgelayer}
		\draw (1.center) to (7.center);
		\draw (5) to (2.center);
		\draw [style=swap] (0.center) to (4.center);
		\draw [style=swap] (6) to (3.center);
		\draw [style=swap, in=-90, out=90, looseness=0.75] (7.center) to (4.center);
		\draw [style=dashed] (9.center) to (8.center);
		\draw [style=dashed] (8.center) to (11.center);
		\draw [style=dashed] (11.center) to (10.center);
		\draw [style=dashed] (9.center) to (10.center);
		\draw (19.center) to (16.center);
		\draw [style=swap] (13.center) to (18.center);
		\draw [style=swap, in=-90, out=90, looseness=0.75] (16.center) to (18.center);
		\draw [style=dashed] (17.center) to (15.center);
		\draw [style=dashed] (15.center) to (12.center);
		\draw [style=dashed] (12.center) to (14.center);
		\draw [style=dashed] (17.center) to (14.center);
		\draw (20.center) to (22.center);
		\draw [style=swap] (21.center) to (23.center);
		\draw [style=swap, in=-90, out=90, looseness=0.75] (22.center) to (23.center);
	\end{pgfonlayer}
\end{tikzpicture}

%% file: diamonddots.tikz
\begin{tikzpicture}
	\begin{pgfonlayer}{nodelayer}
		\node [style=none] (0) at (3, -3.25) {};
		\node [style=none] (1) at (3, -0.5) {};
		\node [style=none] (2) at (1.5, -0.5) {};
		\node [style=none] (3) at (-0.75, -1.75) {};
		\node [style=none] (4) at (0.75, -1.75) {};
		\node [style=none] (5) at (-1.5, -0.5) {};
		\node [style=medium box] (6) at (0, -2.25) {$f_\bot$};
		\node [style=medium box] (7) at (0, 2.25) {$f_\top$};
		\node [style=none] (8) at (1.5, 0.5) {};
		\node [style=none] (9) at (0.75, 1.75) {};
		\node [style=medium box] (10) at (-2.25, 0) {$f_a$};
		\node [style=medium box] (11) at (2.25, 0) {$f_b$};
		\node [style=none] (12) at (-0.75, 1.75) {};
		\node [style=none] (13) at (-1.5, 0.5) {};
		\node [style=none] (14) at (3, 3.25) {};
		\node [style=none] (15) at (3, 0.5) {};
		\node [style=none] (16) at (-3, -0.5) {};
		\node [style=none] (17) at (-3, -3.25) {};
		\node [style=none] (18) at (-3, 3.25) {};
		\node [style=none] (19) at (-3, 0.5) {}; 
		\node [style=none] (20) at (3.5, 3) {};
		\node [style=none] (21) at (3.5, -3) {};
		\node [style=none] (22) at (-3.5, 3) {};
		\node [style=none] (23) at (-3.5, -3) {};
	\end{pgfonlayer}
	\begin{pgfonlayer}{edgelayer}
		\draw (0.center) to (1.center);
		\draw [in=-90, out=90, looseness=1.00] (4.center) to (2.center);
		\draw [in=-90, out=90, looseness=1.00] (3.center) to (5.center);
		\draw [in=-90, out=90, looseness=1.00] (8.center) to (9.center);
		\draw [in=-90, out=90, looseness=1.00] (13.center) to (12.center);
		\draw (15.center) to (14.center);
		\draw (17.center) to (16.center);
		\draw (19.center) to (18.center);
		\draw [style=dashed] (22.center) to (20.center);
		\draw [style=dashed] (20.center) to (21.center);
		\draw [style=dashed] (21.center) to (23.center);
		\draw [style=dashed] (22.center) to (23.center);
	\end{pgfonlayer}
\end{tikzpicture}